\documentclass[runningheads,11pt]{llncs}

\usepackage{geometry}
\usepackage{amsfonts}
\usepackage[utf8]{inputenc}
\usepackage[T1]{fontenc}
\usepackage{lmodern}
\usepackage{amsmath,mathtools}
\usepackage{amssymb}
\usepackage{hyperref}
\usepackage{cite}
\usepackage{mathalfa}
\numberwithin{equation}{section}

\usepackage{bm}
\usepackage{spverbatim}
\usepackage{breakcites}
\usepackage{enumitem}
\usepackage{complexity}
\hypersetup{
	colorlinks,
	citecolor=black,
	filecolor=black,
	linkcolor=black,
	urlcolor=black
}
\renewcommand{\qed}{\hfill \blacksquare}
\usepackage{nccmath}
\usepackage{multicol,xparse}
\NewDocumentCommand{\eulerian}{omm}
{%
	\genfrac<>{0pt}{}{#2}{#3}%
	\IfValueT{#1}{_{\!#1}}%
}

\newcommand{\Stirling}[2]{\genfrac{\{}{\}}{0pt}{}{#1}{#2}}

    \makeatletter
\def\@fnsymbol#1{\ensuremath{\ifcase#1\or \dagger\or \ddagger\or
		\mathsection\or \mathparagraph\or \|\or **\or \dagger\dagger
		\or \ddagger\ddagger \else\@ctrerr\fi}}
\makeatother

\newcommand*\samethanks[1][\value{footnote}]{\footnotemark[#1]}

\title{Access Structure Hiding Secret Sharing \\from Novel Set Systems and Vector Families\thanks{This is the full version of the paper that appears in D. Kim et al. (Eds.): COCOON 2020, LNCS 12273, pp. 246-261. DOI: {10.1007/978-3-030-58150-3\_20}. This version contains tighter bounds on the (maximum) share size, and the total number of access structures supported.}}
\author{Vipin Singh Sehrawat \inst{1}\thanks{Work done while the author was a PhD candidate at The University of Texas at Dallas, USA.}\thanks{Research partially supported by NPRP award NPRP8-2158-1-423 from the Qatar National Research Fund (a member of The Qatar Foundation). The statements made herein are solely the responsibility of the authors.} \and Yvo Desmedt \inst{2,3}\samethanks}
\institute{Seagate Technology, Singapore \\ \email{\{vipin.sehrawat.cs@gmail.com\}} \and 
	The University of Texas at Dallas, Richardson, USA \and
	University College London, London, UK}

\titlerunning{Access Structure Hiding Secret Sharing from Novel Set Systems and Vector Families}
\authorrunning{V. S. Sehrawat and Y. Desmedt}

\begin{document}
\maketitle

\begin{abstract}
	\normalsize
	Secret sharing provides a means to distribute shares of a secret such that any authorized subset of shares, specified by an access structure, can be pooled together to recompute the secret. The standard secret sharing model requires public access structures, which violates privacy and facilitates the adversary by revealing high-value targets. In this paper, we address this shortcoming by introducing \emph{hidden access structures}, which remain secret until some authorized subset of parties collaborate. The central piece of this work is the construction of a set-system $\mathcal{H}$ with strictly greater than $\exp\left(c \dfrac{1.5 (\log h)^2}{\log \log h}\right)$ subsets of a set of $h$ elements. Our set-system $\mathcal{H}$ is defined over $\mathbb{Z}_m$, where $m$ is a non-prime-power, such that the size of each set in $\mathcal{H}$ is divisible by $m$ but the sizes of their pairwise intersections are not divisible by $m$, unless one set is a subset of another. We derive a vector family $\mathcal{V}$ from $\mathcal{H}$ such that superset-subset relationships in $\mathcal{H}$ are represented by inner products in $\mathcal{V}$. We use $\mathcal{V}$ to ``encode'' the access structures and thereby develop the first \emph{access structure hiding} secret sharing scheme. For a setting with $\ell$ parties, our scheme supports $2^{\binom{\ell}{\ell/2+1}}$ out of the $2^{2^{\ell - O(\log \ell)}}$ total monotone access structures, and its maximum share size for any access structures is $(1+ o(1)) \dfrac{2^{\ell+1}}{\sqrt{\pi \ell/2}}$. The scheme assumes semi-honest polynomial-time parties, and its security relies on the Generalized Diffie-Hellman assumption.
\end{abstract}

\section{Introduction}
A secret sharing scheme~\cite{Shamir[79],Blakley[79],Ito[87]} is a method by which a dealer, holding a secret string, distributes strings, called shares, to parties such that authorized subsets of parties, specified by a public access structure, can reconstruct the secret. Secret sharing is the foundation of multiple cryptographic tools (in addition to its obvious use in secure storage), including threhsold cryptography~\cite{Yvo[89]}, (secure) multiparty computation~\cite{Micali[91]}, attribute-based encryption~\cite{Goyal[06]}, generalized oblivious transfer~\cite{Tassa[11]}, perfectly secure message transmission~\cite{Danny[93]}, anonymous communications~\cite{Sehrawat[17]}, e-voting~\cite{Berry[99],Yung[04]} and e-auctions~\cite{Mic[98],Peter[09]}. The extensive survey by Beimel~\cite{Beimel[11]} gives a review of the notable results in the area. 

The maximum share size in the original secret sharing schemes for general/arbitrary access structures~\cite{Ito[87]} is $2^{\ell-1}$, where $\ell$ is the total number of parties. While for specific access structures, the share size of the later schemes~\cite{Brick[89],Karch[93],Gusta[88]} is less than the share size for the scheme from~\cite{Ito[87]}, the share size of all schemes for general access structures remained $2^{\ell-o(\ell)}$ until 2018. In 2018, Liu and Vaikuntanathan~\cite{LiuStoc[18]} (using results from~\cite{Liu[18]}) constructed a secret sharing scheme for general access structures with share size $2^{0.944 \ell}$. Applebaum et al.~\cite{Benny[20]} (using results of~\cite{Apple[19],Liu[18]}) improved those results, and constructed a secret sharing scheme for general access structures with share size $2^{0.637\ell + o(\ell)}$. Whether the share size can be improved to $2^{o(\ell)}$ (or even smaller) remains an important open problem. On the other hand, multiple works~\cite{Blundo[92],Capo[93],Csi[96],Csi[97],Dijk[95]} have proved various lower bounds for secret sharing with the best being $\mathrm{\Omega}(\ell^2/\log \ell)$ from Csirmaz~\cite{Csi[96]}. 

\subsection{Motivation}
Existing secret sharing model requires the access structure to be known to the parties. Since secret reconstruction requires shares of any authorized subset, from the access structure, having a public access structure reveals the high-value targets, which can lead to compromised security in the presence of malicious parties. Having a public access structure also implies that some parties must publicly consent to the fact that they themselves are not trusted. \\[2mm] 
\textit{Need for Hidden Access Structures:}
Consider a scenario where Alice dictates her will/testament and instructs her lawyer that each of her $15$ family members should receive a valid ``share'' of the will. In addition, the shares should be indistinguishable from each other in terms of size and entropy. She also insists that in order to reconstruct her will, \{Bob, Tom, Catherine\} or \{Bob, Cristine, Keri, Roger\} or \{Rob, Eve\} must be part of the collaborating set. But, Alice does not want to be in the bad books of her other, less trusted family members. So, she demands that the shares of her will and the procedure to reconstruct it back from the shares must not reveal her ``trust structures'', until after the will is successfully reconstructed. This problem can be generalized to secret sharing, but with \textit{hidden access structures}, which remain secret until some authorized subset of parties collaborate.\\[2mm] 
\textit{Superpolynomial Size Set-Systems and Efficient Cryptography:}
In this paper, we demonstrate that set-systems with specific intersections can be used to enhance existing cryptographic protocols, particularly the ones meant for distributed security. In order to minimize the computational cost of cryptographic protocols, it is desirable that parameters such as exponents, moduli and dimensions do not grow too big. For a set-system whose size is superpolynomial in the number of elements over which it is defined, achieving a large enough size requires smaller modulus and fewer number of elements, which translates into smaller dimensions, exponents and moduli for its cryptographic applications.

\subsection{Related Work}
A limited number of attempts have been made to introduce privacy-preserving features to secret sharing. The first solution that focused on bolstering privacy in secret sharing was called anonymous secret sharing~\cite{Stinson[87]}, wherein the secret can be reconstructed without the knowledge of which participants hold which shares. In such schemes, secret reconstruction can be carried out by giving the shares to a black box that does not know the identities of the participants holding those shares. However, anonymous secret sharing completely discards parties' identities, which limits its applicability as an extension of secret sharing. Another issue is that the known anonymous secret sharing schemes~\cite{Stinson[87],Phillips[92],Blundo[96],Kishi[02],Deng[07]} operate in very restricted settings (e.g. $n$-out-of-$n$ threshold, $2$-out-of-$n$ threshold) or use hard to generate underlying primitives. For instance, the constructions from~\cite{Stinson[87],Blundo[96]} use resolvable Steiner systems~\cite{Steiner[53]}. However, in design theory, resolvable Steiner systems are non-trivial to achieve with a few known results in restricted settings~\cite{Bryant[17],Teir[94],Col[92],Lou[20],Kwan[20],Pat[17],Pipp[89],Chau[71],Fer[19],Kra[95],Yuca[99],Yuca[02]}. There are also known impossibility results concerning existence of certain desirable Steiner systems~\cite{Pat[08]}. Other attempts made to realize anonymous secret sharing avoided the hard to generate primitives and instead employed combinatorics~\cite{Kishi[02]}. But, they also lead to very restricted and specific thresholds. 

\begin{remark}
Steiner systems have strong connections to a wide range of topics, including statistics, finite group theory, finite geometry, combinatorial design, experimental design, storage systems design, wireless communication, low-density parity-check code design, distributed storage, batch codes, and low-redundancy private information retrieval. For an introduction to the subject, we refer the interested reader to~\cite{Tri[99],Charles[06]}. 
\end{remark}

\subsection{Our Contributions}
We bolster the privacy guarantees of secret sharing by introducing \textit{hidden access structures}, which remain unknown until some authorized subset of parties collaborate. We develop the first access structure hiding (computational) secret sharing scheme. As the basis of our scheme, we construct a novel set-system, which is defined by the following theorem. 
\begin{theorem}~\label{lemma}
	Let $\{\alpha_i\}_{i=1}^r$ be $r > 1$ positive integers and $m = \prod_{i=1}^{r} p_i^{\alpha_i}$ be a positive integer with $r$ different prime divisors: $p_1, \dots, p_r$. Then there exists $c = c(m) > 0$, such that for every integer $h > 0$, there exists an explicitly constructible non-uniform\footnote{all member sets do not have equal size} set-system $\mathcal{H}$ over a universe of $h$ elements such that the following conditions hold:
	\begin{enumerate}
		\item \label{T1} $|\mathcal{H}| > \exp\left(c \dfrac{1.5 (\log h)^r}{(\log \log h)^{r-1}}\right),$
		\item \label{T2} $\forall H \in \mathcal{H}: |H| = 0 \bmod m$,
		\item \label{T3} $\forall G, H \in \mathcal{H},$ where $G \neq H:$ if $H \subset G$ or $G \subset H$, then $|G \cap H| = 0 \bmod m$, else $|G \cap H| \not= 0 \bmod m$,
		\item \label{T4} $\forall G, H \in \mathcal{H},$ where $G \neq H$ and $\forall i \in \{1, \dots, r\} : |G \cap H| \in \{0,1\} \bmod p_i^{\alpha_i}$.
	\end{enumerate}
\end{theorem}

(Recall that $a \bmod m$ denotes the smallest non-negative $b = a \bmod m$.) In secret sharing, the family of minimal authorized subsets $\Gamma_0 \in \Gamma$, corresponding to an access structure $\Gamma$, is defined as the collection of the minimal sets in $\Gamma$. Therefore, $\Gamma_0$ forms the \textit{basis} of $\Gamma$. Note that Conditions~\ref{T2} and~\ref{T3} of Theorem~\ref{lemma} define the superset-subset relations in the set-system $\mathcal{H}$. We derive a family of vectors $\mathcal{V} \in (\mathbb{Z}_m)^h$ from our set-system $\mathcal{H}$, that captures the superset-subset relations in $\mathcal{H}$ as (vector) inner products in $\mathcal{V}$. This capability allows us to capture ``information'' about any minimal authorized subset $\mathcal{A} \in \Gamma_0$ in the form of an inner product, enabling efficient testing of whether a given subset of parties $\mathcal{B}$ is a superset of $\mathcal{A}$ or not. Since $\Gamma$ is monotone, $\mathcal{B} \supseteq \mathcal{A}$, for some $\mathcal{A} \in \Gamma_0$, implies that $\mathcal{B} \in \Gamma$, i.e., $\mathcal{B}$ is an authorized subset of parties. Similarly, $\mathcal{B} \not\supseteq \mathcal{A}$, for all $\mathcal{A} \in \Gamma_0$, implies that $\mathcal{B} \notin \Gamma$, i.e., $\mathcal{B}$ is not an authorized subset of parties. We use our novel set-system and vector family to construct the first access structure hiding (computational) secret sharing scheme. We assume semi-honest polynomial-time parties, and reduce the security and privacy guarantees of our scheme to the Generalized Diffie-Hellman assumption~\cite{Steiner[96]}. For a setting with $\ell$ parties, our scheme supports $2^{\binom{\ell}{\ell/2+1}}$ out of the $2^{2^{\ell - O(\log \ell)}}$ total monotone access structures, and its maximum share size for any of those $2^{\binom{\ell}{\ell/2+1}}$ monotone access structures is $(1+ o(1)) \dfrac{2^{\ell+1}}{\sqrt{\pi \ell/2}}$. Hence, the maximum share size for our access structure hiding secret sharing scheme is greater than the current upper bound of $2^{0.637\ell + o(\ell)}$~\cite{Benny[20]} on the share size for secret sharing schemes for general access structures.

\subsection{Organization} 
The rest of this paper is organized as follows: we recall the pertinent background and results in Section~\ref{Sec2}. Section~\ref{Sec3} formally defines access structure hiding computational secret sharing. We present the construction of our set-systems and vector families in Section~\ref{construction}, and use them to develop the first access structure hiding computational secret sharing scheme in Section~\ref{Sec4}. We conclude by describing two open problems.

\section{Preliminaries}\label{Sec2}
We begin by recalling an informal definition of the Generalized Diffie-Hellman (GDH) assumption~\cite{Steiner[96]}. For a formal definition, see~\cite{Boneh[99]}. For a positive integer $n$, we define $[n]:= \{1,\dots,n\}$.
\begin{definition}[GDH Assumption: Informal]\label{GDH}
		\emph{Let $\{a_1, a_2, \dots, a_n\}$ be a set of $n$ different integers. Given a group $G$ and an element $g \in G$, it is hard to compute $g^{\prod_{i \in [n]}a_i}$ for an algorithm that can query $g^{\prod_{i \in I}a_i}$ for any proper subset $I \subsetneq [n].$}
\end{definition}

\begin{definition}[Dirichlet's Theorem~\cite{Diri[37]}]\label{Dr}
	\emph{For all coprime integers $a$ and $q$, there are infinitely many primes, $p$, of the form $p = a \bmod q.$}
\end{definition}

\begin{definition}[Euler's Theorem]\label{Euler}
	\emph{Let $y$ be a positive integer and $\mathbb{Z}_y^*$ denote the multiplicative group $\bmod ~y$. Then for every integer $c$ that is coprime to $y$, it holds that: $c^{\varphi(y)} = 1 \bmod y,$ where $\varphi(y) = |\mathbb{Z}_y^*|$ denotes Euler's totient function.} 
\end{definition}

\begin{definition}[Hadamard/Schur product]\label{Hada}
	\emph{Hadamard/Schur product of two vectors $\textbf{u}, \textbf{v} \in \mathcal{R}^n$, denoted by $\textbf{u} \circ \textbf{v}$, is a vector in the same linear space whose $i$-th element is defined as: $(\textbf{u} \circ \textbf{v})[i] = \textbf{u}[i] \cdot \textbf{v}[i],$ for all $i \in [n].$}
\end{definition}

\begin{definition}[Negligible Function]
	\emph{For security parameter $\omega$, a function $\epsilon(\omega)$ is called \textit{negligible} if for all $c > 0$ there exists a $\omega_0$ such that $\epsilon(\omega) < 1/\omega^c$ for all $\omega > \omega_0$.}
\end{definition}

\begin{definition}[Computational Indistinguishability~\cite{Gold[82]}]
	\emph{Let $X = \{X_\omega\}_{\omega \in \mathbb{N}}$ and $Y = \{Y_\omega\}_{\omega \in \mathbb{N}}$ be ensembles, where $X_\omega$'s and $Y_\omega$'s are probability distributions over $\{0,1\}^{\kappa(\omega)}$ for some polynomial $\kappa(\omega)$. We say that $\{X_\omega\}_{\omega \in \mathbb{N}}$ and $\{Y_\omega\}_{\omega \in \mathbb{N}}$ are polynomially/computationally indistinguishable if the following holds for every (probabilistic) polynomial-time algorithm $\mathcal{D}$ and all $\omega \in \mathbb{N}$:
	\[\Big| \text{Pr}[t \leftarrow X_\omega: \mathcal{D}(t) = 1] - \text{Pr}[t \leftarrow Y_\omega: \mathcal{D}(t) = 1] \Big| \leq \epsilon(\omega), \]
	where $\epsilon$ is a negligible function.}
\end{definition}	

\begin{definition}[Access Structure] \emph{Let $\mathcal{P} = \{P_1, \dots, P_\ell\}$ be a set of parties. A collection $\Gamma \subseteq 2^{\mathcal{P}}$ is monotone if $\mathcal{A} \in \Gamma$ and $\mathcal{A} \subseteq \mathcal{B}$ imply that $\mathcal{B} \in \Gamma$. An access structure $\Gamma \subseteq 2^{\mathcal{P}}$ is a monotone collection of non-empty subsets of $\mathcal{P}$. Sets in $\Gamma$ are called authorized, and sets not in $\Gamma$ are called unauthorized.}
\end{definition}

If $\Gamma$ consists of all subsets of $\mathcal{P}$ with size greater than or equal to a fixed threshold $t$ $(1 \leq t \leq \ell)$, then $\Gamma$ is called a $t$-threshold access structure. 

\begin{definition}[Minimal Authorized Subset]\label{GammaDef}
	\emph{For an access structure $\Gamma$, a family of minimal authorized subsets $\Gamma_0 \in \Gamma$ is defined as:
	\[
		\Gamma_0 = \{\mathcal{A} \in \Gamma: \mathcal{B} \not\subset \mathcal{A} \text{ for all } \mathcal{B} \in \Gamma \setminus \{ \mathcal{A} \}\}.
	\]  }   
\end{definition}

\begin{definition}[Computational Secret Sharing~\cite{Hugo[93]}]\label{def.1}
	\emph{A computational secret sharing scheme with respect to an access structure $\Gamma$, security parameter $\omega$, a set of $\ell$ polynomial-time parties $\mathcal{P} = \{P_1, \dots, P_\ell \}$, and a set of secrets $\mathcal{K}$, consists of a pair of polynomial-time algorithms, {\fontfamily{cmtt}\selectfont (Share,Recon)}, where: 
	\begin{itemize}
		\item {\fontfamily{cmtt}\selectfont Share} is a randomized algorithm that gets a secret $k \in \mathcal{K}$ and access structure $\Gamma$ as inputs, and outputs $\ell$ shares, $\{\mathrm{\Pi}^{(k)}_1, \dots, \mathrm{\Pi}^{(k)}_\ell\},$ of $k$,
		\item {\fontfamily{cmtt}\selectfont Recon} is a deterministic algorithm that gets as input the shares of a subset $\mathcal{A} \subseteq \mathcal{P}$, denoted by $\{\mathrm{\Pi}^{(k)}_i\}_{i \in \mathcal{A}}$, and outputs a string in $\mathcal{K}$,
	\end{itemize}
	such that, the following two requirements are satisfied:
	\begin{enumerate}
		\item \textit{Perfect Correctness:} for all secrets $k \in \mathcal{K}$ and every authorized subset $\mathcal{A} \in \Gamma$, it holds that: \\ Pr[{\fontfamily{cmtt}\selectfont Recon}$(\{\mathrm{\Pi}^{(k)}_i\}_{i \in \mathcal{A}}, \mathcal{A}) = k] = 1,$ 
		\item \textit{Computational Secrecy:} for every unauthorized subset $\mathcal{B} \notin \Gamma$ and all different secrets $k_1, k_2 \in \mathcal{K}$, it holds that the distributions $\{\mathrm{\Pi}_i^{(k_1)}\}_{i \in \mathcal{B}}$ and $\{\mathrm{\Pi}_i^{(k_2)}\}_{i \in \mathcal{B}}$ are computationally indistinguishable (with respect to $\omega)$.
	\end{enumerate}}
\end{definition}

\begin{remark}[Perfect Secrecy]\label{remark}
	If $\forall k_1, k_2 \in \mathcal{K}$ with $k_1 \neq k_2$, the distributions $\{\mathrm{\Pi}_i^{(k_1)}\}_{i \in \mathcal{B}}$ and $\{\mathrm{\Pi}_i^{(k_2)}\}_{i \in \mathcal{B}}$ are identical, then the scheme is called a perfect secret sharing scheme.
\end{remark}

\subsection{Set Systems with Restricted Intersections}
The  problem  of constructing  set  systems  under  certain  intersection  restrictions and bounding their size  has a  central  place  in  Extremal  Set  Theory.   We  shall  not  give  a  full  account  of  such  problems, but only touch upon the results that are particularly relevant to our set-system and its construction. For a broader account, we refer the interested reader to the survey by Frankl and Tokushige~\cite{Frankl[16]}.  

\begin{lemma}[\cite{Gro[00]}]\label{cor2}
	Let $m = \prod_{i=1}^{r} p_i^{\alpha_i}$ be a positive integer with $r > 1$ different prime divisors. Then there exists an explicitly constructible polynomial $Q$ with $n$ variables and degree $O(n^{1/r})$, which is equal to $0$ on $z = (1,1, \dots, 1) \in \{0,1\}^n$ but is nonzero $\bmod~ m$ on all other $z \in \{0,1\}^n$. Furthermore, $\forall z \in \{0,1\}^n$ and $\forall i \in \{1, \dots ,r\}$, it holds that: $Q(z) \in \{0,1\} \bmod p_i^{\alpha_i}$.
\end{lemma}

\begin{theorem}[\cite{Gro[00]}]\label{thm}
	Let $m$ be a positive integer, and suppose that $m$ has $r > 1$ different prime divisors: $m = \prod_{i=1}^{r} p_i^{\alpha_i}$. Then there exists $c = c(m) > 0$, such that for every integer $h > 0$, there exists an explicitly constructible uniform set-system $\mathcal{H}$ over a universe of $h$ elements such that:
	\begin{enumerate}
		\item $|\mathcal{H}| \geq \exp \left( c \dfrac{(\log h)^r}{(\log \log h)^{r-1}} \right)$,
		\item $\forall H \in \mathcal{H}:|H| = 0 \bmod m$,
		\item $\forall G, H \in \mathcal{H}, G \neq H:|G \cap H| \not= 0 \bmod m$.
	\end{enumerate}
\end{theorem}

\subsubsection{Matching Vectors.} 
A matching vector family is a combinatorial object that is defined as:

\begin{definition}[\cite{Zeev[11]}]
	\emph{Let $S \subseteq \mathbb{Z}_m \setminus \{0\}$, and $\langle \cdot, \cdot \rangle$ denote the inner product. We say that subsets $\mathcal{U} = \{\textbf{u}_i\}_{i=1}^N$ and $\mathcal{V} = \{\textbf{v}_i\}_{i=1}^N$ of vectors in $(\mathbb{Z}_m)^h$ form an $S$-matching family if the following two conditions are satisfied: 
		\begin{itemize}
			\item $\forall i \in [N],$ it holds that: $\langle \textbf{u}_i, \textbf{v}_i \rangle = 0 \bmod m$, 
			\item $\forall i,j \in [N]$ such that $i \neq j$, it holds that: $\langle \textbf{u}_i, \textbf{v}_j \rangle \bmod m \in S$.
	\end{itemize}}
\end{definition}

The question of bounding the size of matching vector families is closely related to the well-known Extremal Set Theory problem of constructing set systems with restricted modular intersections. Matching vectors have found applications in the context of private information retrieval~\cite{Beimel[15],Beimel[12],Zeev[15],Zeev[11],Klim[09],Sergey[08],Liu[17]}, conditional disclosure of secrets~\cite{Liu[17]}, secret sharing~\cite{Liu[18]} and coding theory~\cite{Zeev[11]}. The first super-polynomial size matching vector family follows directly from the set-system constructed by Grolmusz~\cite{Gro[00]}. If each set $H$ in the set-system $\mathcal{H}$ defined by Theorem~\ref{thm} is represented by a vector $\textbf{u} \in (\mathbb{Z}_m)^h$, then it leads to the following family of $S$-matching vectors:
 
\begin{corollary}[\cite{Zeev[11]} to Theorem~\ref{thm}]
	Let $h, m > 0$, and suppose that $m = \prod_{i=1}^{r} p_i^{\alpha_i}$ has $r > 1$ different prime divisors. Then, there exists a set $S$ of size $2^{r-1}$ and a family of $S$-matching vectors \emph{$\{\textbf{u}\}$}${}^N_{i=1},$ \emph{$\textbf{u}_i$} $\in (\mathbb{Z}_m)^h$, such that, $N \geq \exp \left( c \dfrac{(\log h)^r}{(\log \log h)^{r-1}} \right)$.
\end{corollary}

\section{Access Structure Hiding Computational Secret Sharing: Definition}\label{Sec3}
In this section, we give a formal definition of an access structure hiding computational secret sharing scheme. 
\begin{definition}\label{MainDef}
	\emph{An access structure hiding computational secret sharing scheme with respect to an access structure $\Gamma$, a set of $\ell$ polynomial-time parties $\mathcal{P} = \{P_1, \dots, P_\ell\}$, a set of secrets $\mathcal{K}$ and a security parameter $\omega$, consists of two pairs of polynomial-time algorithms, {\fontfamily{cmtt}\selectfont(HsGen, HsVer)} and {\fontfamily{cmtt}\selectfont(Share, Recon)}, where {\fontfamily{cmtt}\selectfont(Share, Recon)} are the same as defined in the definition of computational secret sharing~(see Definition~\ref{def.1}), and {\fontfamily{cmtt}\selectfont(HsGen, HsVer)} are defined as:
		\begin{itemize}
			\item {\fontfamily{cmtt}\selectfont HsGen} is a randomized algorithm that gets $\mathcal{P}$ and $\Gamma$ as inputs, and outputs $\ell$ \textit{access structure tokens} $\{\mathrm{\mho}^{(\Gamma)}_1, \dots, \mathrm{\mho}^{(\Gamma)}_\ell\},$ 
			\item {\fontfamily{cmtt}\selectfont HsVer} is a deterministic algorithm that gets as input the \textit{access structure tokens} of a subset $\mathcal{A} \subseteq \mathcal{P}$, denoted by $\{\mathrm{\mho}_i^{(\Gamma)}\}_{i \in \mathcal{A}}$, and outputs $b \in \{0,1\}$,
		\end{itemize}
		such that, the following three requirements are satisfied:
		\begin{enumerate}
			\item \textit{Perfect Completeness:} every authorized subset of parties $\mathcal{A} \in \Gamma$ can identify itself to be a member of the access structure $\Gamma$, i.e., formally, it holds that: $Pr[${\fontfamily{cmtt}\selectfont HsVer}$(\{\mathrm{\mho}_i^{(\Gamma)}\}_{i \in \mathcal{A}}) = 1] = 1,$
			\item \textit{Perfect Soundness:} every unauthorized subset of parties $\mathcal{B} \notin \Gamma$ can identify itself to be outside of the access structure $\Gamma$, i.e., formally, it holds that: $Pr[${\fontfamily{cmtt}\selectfont HsVer}$(\{\mathrm{\mho}_i^{(\Gamma)}\}_{i \in \mathcal{B}}) = 0] = 1,$
			\item \textit{Computational Hiding:} for all access structures $\Gamma, \Gamma' \subseteq 2^{\mathcal{P}}$, where $\Gamma \neq \Gamma'$, and each subset of parties $\mathcal{B} \notin \Gamma, \Gamma'$ that is unauthorized in both $\Gamma$ and $\Gamma'$, it holds that:
			\[\left| Pr[\Gamma~|~ \{\mathrm{\mho}_i^{(\Gamma)}\}_{i \in \mathcal{B}}, \{\mathrm{\Pi}_i^{(k)}\}_{i \in \mathcal{B}}] - Pr[\Gamma'~|~ \{\mathrm{\mho}_i^{(\Gamma)}\}_{i \in \mathcal{B}}, \{\mathrm{\Pi}_i^{(k)}\}_{i \in \mathcal{B}}] \right| \leq \epsilon(\omega),\]
			where $\epsilon$ is a negligible function and $\{\mathrm{\Pi}_i^{(k)}\}_{i \in \mathcal{B}}$ denotes the subset of shares of a secret $k$, that belong to the parties in $\mathcal{B}$, and are generated by {\fontfamily{cmtt}\selectfont Share} with respect to the access structure $\Gamma$.
	\end{enumerate}}
\end{definition}

\section{Novel Set-Systems and Vector Families}\label{construction}
In this section, we construct our novel set-systems and vector families. The following notations are frequently used throughout this section.
\begin{itemize} 
	\item We denote the coefficient of $x^k$ in the power series for $f(x)$ by $[x^k]: f(x)$,
	\item Let $L$ be an ordered list of a finite number of different symbols, and $u \in L^e$ be a string comprised of $e \in \mathbb{N}$ different symbols from $L$. We define $\rhd$ to represent \textit{string membership}, i.e., $j \rhd u$ denotes that the string $u$ contains the $j^{th}$ symbol from the ordered list $L$.  
\end{itemize}

\subsection{Set System Construction}
In this section, we provide the proof for Theorem~\ref{lemma} by giving an explicit construction of the set-system $\mathcal{H}$ defined in it. Our construction is inspired by that of Grolmusz~\cite{Gro[00]}.
\begin{proof}[Theorem~\ref{lemma}]
	We use the polynomial $Q$ defined in Lemma~\ref{cor2} to construct our set-system. We begin by recalling the following property of $Q$: 
	\begin{equation}\label{eqn1}
	Q(z) = 0 \bmod m \Longleftrightarrow z_1 = z_2 = \dots = z_n = 1,
	\end{equation}
	where $z = (z_1, z_2, \dots z_n) \in \{0,1\}^n$. We know from Lemma~\ref{cor2} that $Q$ has degree $d = O(n^{1/r})$, and can be written as:
	\[Q(z_1,z_2,\dots,z_n) = \sum\limits_{i_1,i_2,\dots,i_l} a_{i_1, i_2, \dots, i_l} z_{i_1} z_{i_2} \dots z_{i_l},\]
	where $0 \leq l \leq d$, and $a_{i_1, i_2, \dots, i_l} \in \mathbb{Z}$ with $1 \leq i_1 < i_2 < \dots < i_l \leq n$. Reducing that modulo $m$, we get:
	\begin{equation}\label{eqn2.2}
	\tilde{Q}(z_1, z_2, \dots, z_n) = \sum\limits_{i_1,i_2,\dots,i_l} \tilde{a}_{i_1, i_2, \dots, i_l} z_{i_1}z_{i_2} \dots z_{i_l},
	\end{equation} 
	where $\tilde{a}_{i_1, i_2, \dots, i_l} = a_{i_1, i_2, \dots, i_l} \bmod m$. Let $L = (0,1,\dots,n-1)$ be an ordered list of $n$ symbols. Define a characteristic function $\psi: \{0, 1, \dots, n-1\}^n \rightarrow \{0,1\}^n$ as:
	\begin{equation}\label{eqn2.3}
	\begin{aligned}
	\psi(u)[j] &:= 
	\begin{cases}
	1 \qquad \qquad \qquad \text{if } j \rhd u \\
	0 \qquad \qquad \qquad \text{otherwise},
	\end{cases}
	\end{aligned}
	\end{equation}
	where $1 \leq j \leq n$ and $\psi(u)[j]$ denotes the $j^{th}$ bit of $\psi(u) \in \{0,1\}^n$. If a string $u \in \{0, 1, \dots, n - 1\}^n$, defined over the symbols in $L$, contains the $j^{th}$ symbol from the ordered list $L$, then $\psi(u)[j] = 1$, else $\psi(u)[j] = 0$. Define a comparison function $\delta(x,y): \{0,1\} \times \{0,1\} \rightarrow \{0,1\}$ as:
	\begin{equation}\label{eqn2.4}
	\delta(u, v) := \neg (u \oplus v),
	\end{equation}
	where $\neg$ and $\oplus$ denote negation and XOR, respectively. Hence, $\delta(u,v) = 1$ if $u = v$, else $\delta(u,v) = 0$. Let $\textbf{A} = (a_{x,y})$ be a $n^n \times n^n$ matrix $(x,y \in \{0, 1, \dots, n-1\}^n)$. For $x' = \psi(x)$ and $y' = \psi(y)$, define each entry $a_{x,y}$ as:
	\begin{align}\label{eqn}
	a_{x,y} = \tilde{Q}(\delta(x'_1, y'_1), \delta(x'_2, y'_2), \dots, \delta(x'_n, y'_n)) \bmod m,
	\end{align}
	where $\tilde{Q}(\cdot)$ is the polynomial defined in Equation~\ref{eqn2.2}, and $x'_j, y'_j$ denote the $j^{th}$ bit of the binary bit strings $x',y' \in \{0,1\}^n$. It follows from Equation~\ref{eqn2.3}, Equation~\ref{eqn2.4} and Equation~\ref{eqn} that if $a_{x,y} = \tilde{Q}(1,1,\dots,1) = 0 \bmod m$, then either $x = y$ or $\forall j \in [n]$ it holds that $y'_j = x'_j$, i.e., $x$ and $y$ are comprised of the same symbols. In both cases, we say that $x$ and $y$ ``cover'' each other, and denote it by $x \mathrm{\Upsilon} y$. We know from Equation~\ref{eqn2.2} that the polynomial $\tilde{Q}(z)$ can be defined as a sum of monomials $z_{i_1}z_{i_2}\dots z_{i_l}~(l\leq d)$, where each monomial $z_{i_1}z_{i_2}\dots z_{i_l}$ occurs with multiplicity $\tilde{a}_{i_1, i_2, \dots, i_l}$ in the sum. Therefore, since matrix $\textbf{A}$ is generated via $\tilde{Q}$, it follows from Equation~\ref{eqn2.2} that $\textbf{A}$ can be defined as the sum of matrices $\textbf{B}_{i_1, i_2, \dots, i_l}$, whose entries are defined as:
	\begin{equation}\label{eqn.2.5}
	b^{i_1, i_2, \dots, i_l}_{x,y} = \delta(x'_{i_1}, y'_{i_1}) \delta(x'_{i_2}, y'_{i_2}) \dots \delta(x'_{i_l}, y'_{i_l}).
	\end{equation}
	Hence, it follows from Equation~\ref{eqn2.2}, Equation~\ref{eqn} and Equation~\ref{eqn.2.5}, that $\textbf{A}$ can be written as:
	\begin{equation}\label{3.7}
		\textbf{A} = \sum\limits_{i_1,i_2,\dots,i_l} \tilde{a}_{i_1,i_2,\dots,i_l} \textbf{B}_{i_1, i_2, \dots, i_l},	
	\end{equation}
	where $\tilde{a}_{i_1,i_2,\dots,i_l}$ is the multiplicity with which the matrix $\textbf{B}_{i_1, i_2, \dots, i_l}$ occurs in the sum. Next, we analyze the matrices $\textbf{A}$ and $\textbf{B}_{i_1, i_2, \dots, i_l}$. In particular, we count the number of $0$ entries in $\textbf{A}$ and the number of $1$ entries in $\textbf{B}_{i_1, i_2, \dots, i_l}$.\\[1.5mm]	
	\textbf{Analysis of the Matrices.}
	We begin by counting the total number of entries $a_{x,y} \in \textbf{A}$ that are equal to $0$, which translates into counting the number of $x,y \in \{0,1,\dots,n-1\}^n$ such that $x \mathrm{\Upsilon} y$. 
	
	Let $\mathcal{S}$ be a set of $n$ different symbols. Let \textit{unique symbol weight} (USW) denote the number of different symbols in a string, i.e., USW$(x) = $w$(\psi(x))$, where w$(\cdot)$ denotes the Hamming weight. To form a string $x$ of length $n$ such that USW$(x) = k$, for a fixed $k \leq n$, the first step is to select $k$ distinct symbols $s_{i_1}, s_{i_2} \dots, s_{i_k}$ from $\mathcal{S}$. We know from Rosen~\cite{Rosen[10]} (Section 2.4.2), that the number of onto functions from a set of $n$ elements to a set of $k$ elements is given by $k! \Stirling{n}{k}$, where $\Stirling{n}{k}$ denotes Stirling number of the second kind (see Graham et al.~\cite{Knuth[94]}, p. 257). Hence, $k! \Stirling{n}{k}$ is the total number of strings of length $n$, that contain exactly the selected $k$-out-of-$n$ symbols: $s_{i_1}, s_{i_2} \dots, s_{i_k}$. 
	
	Let $N_k$ denote the total number of different $x \in \{0,1,\dots,n-1\}^n$ such that USW$(x) = k$. We know that for a fixed set of $k$-out-of-$n$ symbols, the number strings $x \in \{0,1,\dots,n-1\}^n$ satisfying USW$(x) = k$ is $k! \Stirling{n}{k}$. Accounting for the number of ways one can choose $k$-out-of-$n$ symbols, we get: $$N_k = \binom{n}{k} k! \Stirling{n}{k}.$$
	
	We know that for each $k$, there are $N_k$ rows in matrix $\textbf{A}$ that ``cover'' exactly $k! \Stirling{n}{k}$ entries. Hence, from Equation~\ref{eqn}, the number of $a_{x,y} = 0 \bmod m$ entries in $\textbf{A}$ is: 
	\begin{equation}\label{sahii}
	S(n) = \sum\limits_{k=1}^n N_k \cdot k! \Stirling{n}{k} = \sum\limits_{k=1}^n \dbinom{n}{k} k! \Stirling{n}{k} k! \Stirling{n}{k}.
	\end{equation}
	We recall the following well known identities involving the first-order Eulerian numbers (see Graham et al.~\cite{Knuth[94]}, p. 267) and Stirling numbers of the second kind:
	\[
	\ell ! \Stirling{n}{\ell} = \sum\limits_{k=0}^n \eulerian{n}{k} \dbinom{k}{n-\ell}; \qquad \quad (n - \ell)! \Stirling{n}{n - \ell} = \sum\limits_{k=0}^n \eulerian{n}{k} \dbinom{k}{\ell},
	\]
	where $\eulerian{n}{k}$ denotes the first-order Eulerian number, which gives the total number of permutations $\pi_1, \pi_2, \dots, \pi_n$ with $k$ \textit{ascents}, i.e., $k$ places where $\pi_t < \pi_{t+1}$. Therefore, Equation~\ref{sahii} can be rewritten as:
	\begin{align*}
	S(n) &=  \sum\limits_{k=0}^n \dbinom{n}{k} k! \Stirling{n}{k} k! \Stirling{n}{k} \\ 
	&= \sum\limits_{k=0}^n \dbinom{n}{k} k! \Stirling{n}{k} \sum\limits_{j=0}^n \eulerian{n}{j} \dbinom{j}{n-k} \\
	&= n! \sum\limits_{k=0}^n \Stirling{n}{k} \left( \sum\limits_{j=0}^n \eulerian{n}{j} \dbinom{j}{n-k} \right) \dfrac{1}{(n-k)!}. 
	\end{align*}
	Thus, the exponential generating function for $S(n)$ comes out to be:
	\[
	\sum\limits_{n \geq 0} S(n) \dfrac{x^n}{n!} = \sum\limits_{n \geq 0} \sum\limits_{k = 0}^n \Stirling{n}{k} x^k \left( \sum\limits_{j=0}^n \eulerian{n}{j} \dbinom{j}{n-k} \right) \dfrac{x^{n-k}}{(n-k)!}. 
	\]
	Recall the following definition of Touchard polynomial (Jacques Touchard~\cite{Touch[39]}):
	\[
	T_n(x) = \sum\limits_{k=0}^n \Stirling{n}{k} x^k.
	\] 
	We write $S(n)$ as: 
	\begin{equation}\label{size}
	S(n) = n![x^n]: (T_n(x)P_n(x)),
	\end{equation}
	where the second polynomial, $P_n(x)$, is defined via convolution as:
	\[
	P_n(x) = \sum\limits_{k=0}^n \left( \sum\limits_{j=0}^n \eulerian{n}{j} \dbinom{j}{k} \right) \dfrac{x^k}{k!} = \sum\limits_{k=0}^n \dfrac{(n-k)!}{k!} \Stirling{n}{n-k} x^k.
	\]
	Observe that all diagonal entries $a_{x,x}$ in matrix $\textbf{A}$ are $0$, and $\textbf{A}$ is symmetric across its diagonal. 
	\begin{lemma}\label{lemma3}
		Let the term \emph{B-entries} denote the entries $b^{i_1,i_2,\dots,i_l}_{x,y} \in \textbf{B}_{i_1,i_2,\dots,i_l}$ that are equal to $1$. Then the following holds for \emph{B-entries:} 
		\begin{enumerate}
			\item\label{1} $\forall x \in \{0,1,\dots,n-1\}^n$, each entry $a_{x,x} \in \textbf{A}$ has the same number of \emph{B-entries}, $b_{x,x}^{i_1,i_2,\dots,i_l} = 1$, and this number is divisible by $m$,
			\item\label{2} for each pair $x,y~(x,y \in \{0,1,\dots, n-1\}^n)$, the total number of \emph{B-entries}, $b_{x,y}^{i_1, i_2, \dots, i_l}= 1$, corresponding to $a_{x,y} \in \textbf{A}$, is divisible by $m$ iff $x \mathrm{\Upsilon} y$, else not.
		\end{enumerate}
	\end{lemma}
	\begin{proof}
		We know from Equation~\ref{eqn.2.5} that except for the B-entries, all other entries in matrices $\textbf{B}_{i_1,i_2,\dots,i_l}$ are equal to $0$. Hence, it follows from Equation~\ref{3.7} that each entry $a_{x,y} \in \textbf{A}$ is simply the total number of B-entries, $b^{i_1,i_2,\dots,i_l}_{x,y} = 1$. It further follows from Equation~\ref{eqn} and Equation~\ref{eqn1} that for all $x$, we get $a_{x,x} = \tilde{Q}(1,1,\dots,1) = 0 \bmod m,$ i.e., for all $x$, the total number of B-entries, $b_{x,x}^{i_1,i_2,\dots,i_l} = 1$, is divisible by $m$. Furthermore, it follows from Equation~\ref{eqn.2.5} that because $x = x$, all entries $b_{x,x}^{i_1,i_2,\dots,i_l}$ are indeed B-entries and all cells $a_{x,x}$ have the same number of corresponding B-entries, $b_{x,x}^{i_1,i_2,\dots,i_l} = 1$. Finally, it follows from Equation~\ref{eqn} and Equation~\ref{eqn1} that for all pairs $(x,y)$, where $x \neq y$, the total number of B-entries, $b^{i_1,i_2,\dots,i_l}_{x,y} = 1$, is: $a_{x,y} = \tilde{Q}(1,1,\dots,1) = 0 \bmod m \text{ if } x \mathrm{\Upsilon} y, \text{ and } a_{x,y} \neq 0 \bmod m$ otherwise. $\qed$ 
	\end{proof}
	
	By taking all $a_{x,y} = 0 \bmod m~(\forall x,y \in \{0,1,\dots, n-1\}^n)$ entries of $\textbf{A}$ to denote sets with the corresponding B-entries, $b^{i_1,i_2,\dots,i_l}_{x,y} = 1$, as the elements in those sets leads to a set-system $\mathcal{H}$, that satisfies Conditions~\ref{T2} and~\ref{T3} of Theorem~\ref{lemma}. The number of elements, $h$, over which $\mathcal{H}$ is defined is:
	\begin{align*}
	h = \tilde{Q}(n,n,\dots,n) &= \sum\limits_{l \leq d} \sum\limits \tilde{a}_{i_1, i_2, \dots, i_l} n^l \leq (m-1) \sum\limits_{l \leq d} \dbinom{n}{l}n^l
	\\	&< (m-1)\sum\limits_{l \leq d}n^{2l}/l! < 2(m-1)n^{2d}/d!,
	\end{align*}
	assuming $n \geq 2d$. Since $d > 2$, we get: $n > 4$. From Equation~\ref{size}, it is easy to verify that the following holds for $n > 2$: 
	\begin{equation}\label{size1}
		|\mathcal{H}| = S(n) > n^{1.5n}.	
	\end{equation}
	We know from~\cite{Gro[00]} that for $r > 1, m = \prod_{i=1}^r p_i^{\alpha_i}, d = O(n^{1/r}), c = c(m) > 0$ and $h < 2(m-1)n^{2d}/d!$, the following relation holds: \[n^n \geq \exp\left(c \dfrac{(\log h)^r}{(\log \log h)^{r-1}}\right).\] 
	Therefore, the following can be derived from Equation~\ref{size1} and elementary estimations for binomial coefficients:
	\[|\mathcal{H}| > \exp\left(c \dfrac{1.5 (\log h)^r}{(\log \log h)^{r-1}}\right).\]
	A tighter bound can be derived by using Lambert $\mathcal{W}$ function~\cite{LambertOrg[58]} and the results from Corless et al.~\cite{Corless[96]} on the principal branch of Lambert $\mathcal{W}$ function, but the bound derived above suffices for our purpose. Since $m \geq 6$ and $r \geq 2$, the size of our set-system $\mathcal{H}$ is strictly greater than $\exp\left(c \dfrac{1.5 (\log h)^2}{\log \log h}\right)$. Condition 4 of Theorem~\ref{lemma} follows directly from Lemma~\ref{cor2}. It is easy to verify that the total number of B-entries corresponding to each cell $(x,y)$, where $x \neq y$ and for which $a_{x,y} = 0 \bmod m$, is not same. Moreover, since all $b_{x,x}^{i_1,i_2,\dots,i_l}$ entries are indeed B-entries, it holds that $a_{x,y} < a_{x,x}$  for all $x \neq y$. Hence, the sets in $\mathcal{H}$ do not have the same size, making $\mathcal{H}$ a non-uniform set-system. This completes the proof of Theorem~\ref{lemma}. $\qed$
\end{proof} 

\subsection{Covering Vector Families}
\begin{definition}[Covering Vectors]\label{def22}
	\emph{Let $m, h > 0$ be positive integers, $S \subseteq \mathbb{Z}_m \setminus \{0\}$, and w$(\cdot)$ and $\langle \cdot, \cdot \rangle$ denote Hamming weight and inner product, respectively. We say that a subset $\mathcal{V} = \{\textbf{v}_i\}_{i=1}^N$ of vectors in $(\mathbb{Z}_m)^h$ forms an $S$-covering family of vectors if the following two conditions are satisfied: 
		\begin{itemize}
			\item $\forall i \in [N]$, it holds that: $\langle \textbf{v}_i, \textbf{v}_i \rangle = 0 \bmod m$,
			\item $\forall i,j \in [N]$, where $i \neq j$, it holds that: 
			\begin{align*}
				\langle \textbf{v}_i, \textbf{v}_j \rangle \bmod m &= 
				\begin{cases}
					0 \qquad \qquad \quad \text{if w}(\textbf{v}_i \circ \textbf{v}_j \bmod m) = 0 \bmod m \\
					\in S \qquad \quad \quad \text{otherwise},
				\end{cases}
			\end{align*}
		\end{itemize}
	where $\circ$ denotes Hadamard/Schur product (see Definition~\ref{Hada}).}
\end{definition}

Recall from Theorem~\ref{lemma} that $h, m$ are positive integers, with $m = \prod_{i=1}^r p_i^{\alpha_i}$ having $r > 1$ different prime divisors. Further, recall Condition~\ref{T4} of Theorem~\ref{lemma}, which implies that the sizes of the pairwise intersections of the sets in $\mathcal{H}$ occupy at most $2^{r}-1$ residue classes modulo $m$. If each set $H_i \in \mathcal{H}$ is represented by a representative vector $\textbf{v}_i \in (\mathbb{Z}_m)^h$, then for the resulting subset $\mathcal{V}$ of vectors in $(\mathbb{Z}_m)^h$, the following result follows from Theorem~\ref{lemma}.   
\begin{corollary}[to Theorem~\ref{lemma}]\label{corImp}
	For the set-system $\mathcal{H}$ defined in Theorem~\ref{lemma}, if each set $H_i \in \mathcal{H}$ is represented by a unique vector \emph{$\textbf{v}_i$} $\in (\mathbb{Z}_m)^h$, then for a set $S$ of size $2^{r}-1,$ the set of vectors $\mathcal{V} =$ \emph{$\{\textbf{v}_i\}$}${}^N_{i=1}$, formed by the representative vectors of all sets in $\mathcal{H}$, forms an $S$-covering family such that $N > \exp\left(c \dfrac{1.5 (\log h)^r}{(\log \log h)^{r-1}}\right)$ and $\forall i,j \in [N]$ it holds that \emph{$\langle \textbf{v}_i, \textbf{v}_j \rangle$}$ = |H_i \cap H_j|(\bmod~ m)$.
\end{corollary}

\section{Our Scheme}\label{Sec4}
In Section~\ref{5.1}, we introduce an algorithm to encode and identify hidden access structures, that remain unknown unless some authorized subset of parties collaborate. Followed by that, in Section~\ref{Full}, we extend that algorithm into an access structure hiding computational secret sharing scheme. We assume semi-honest polynomial-time parties, which try to gain additional information while correctly following the protocols. The following notations are frequently used from hereon.
\begin{itemize}
	\item If each party $P_i$ holds a value $x_i$, then for any subset of parties $\mathcal{A}$, $\{x_i\}_{i \in \mathcal{A}}$ denotes the set of all $x_i$ values that belong to the parties $P_i \in \mathcal{A}$,
	\item $\prod_{i \in \mathcal{A}} x_i$ and $\sum_{i \in \mathcal{A}} x_i$ respectively denote the product and sum of all values from the set $\{x_i\}_{i \in \mathcal{A}}$,
	\item large prime: refers to a prime number of size equal to or greater than the minimum size recommended by NIST for primes~\cite{NIST[18]}.
\end{itemize}

\subsection{Access Structure Encoding Scheme (ASES)}\label{5.1}
In this section, we describe our scheme to encode and identify hidden access structures. Let $\mathcal{P} = \{P_1, \dots, P_\ell\}$ be a set of $\ell$ polynomial-time parties and $\mathrm{\Omega} \in \Gamma_0$ be any minimal authorized subset (see Definition~\ref{GammaDef}). Hence, each party $P_i \in \mathcal{P}$ can be identified as $P_i \in \mathrm{\Omega}$ or $P_i \in \mathcal{P} \setminus \mathrm{\Omega}$. \\[2mm]
\textbf{Setup.} The scheme is initialized as follows:
\begin{enumerate}
	\item For $\eta \geq \ell$, generate a set of distinct large primes, $\{p_1, p_2, \dots, p_{\eta}\}$. Generate a prime $q = u \prod_{i=1}^\eta p_i + 1$, where $u$ is an integer. We know from Dirichlet's Theorem (see Definition~\ref{Dr}) that there are infinitely many such primes $q$. Generating $q$ in this manner ensures hardness of the discrete log problem in $\mathbb{Z}_q$~\cite{Ant[14]} which, by extension, translates into hardness of the Generalized Diffie-Hellman assumption in $\mathbb{Z}_q$.	
	\item Let $w = \prod_{i=1}^\eta p_i$ and $m = \varphi(q)$. Then, it follows from $q = u \prod_{i=1}^\eta p_i + 1$ that $w | \varphi(q)$, where $\varphi$ denotes Euler's totient function (see Definition~\ref{Euler}). Hence, the following holds for $d \geq 1$ primes $\beta_d$ and positive integers $\alpha_d$: \[m = w \cdot \prod\limits_{d \geq 1} \beta_d^{\alpha_d} = \prod\limits_{i=1}^\eta p_i \cdot \prod\limits_{d \geq 1} \beta_d^{\alpha_d}.\]  
	Let $r = d + \eta > \ell$ denote the total number of prime factors of $m$. 
	\item Construct a set-system $\mathcal{H}$ modulo $m$ (as defined by Theorem~\ref{lemma}). Let $\mathcal{V} \in (\mathbb{Z}_m)^h$ denote the covering vectors family (as defined by Corollary~\ref{corImp}) representing $\mathcal{H}$ such that each vector $\textbf{v}_i \in \mathcal{V}$ represents a unique set $H_i \in \mathcal{H}$.
	\item Randomly sample $H \in \mathcal{H}$. Let $\textbf{v} \in \mathcal{V}$ be the representative vector for $H$. We call $\textbf{v}$ and $H$ the \textit{access structure vector} and \textit{access structure set}, respectively.
\end{enumerate}
\textbf{Distributing Access Structures.} Following procedure ``encodes'' the access structure $\Gamma$ that originates from $\mathrm{\Omega}$, and outputs $\ell$ \textit{access structure tokens}. 
\begin{enumerate}
	\item For each party $P_i \in \mathrm{\Omega}$, randomly select a unique vector $\textbf{v}_i \xleftarrow{\: \$ \:} \mathcal{V}$, such that, $\langle \textbf{v}, \textbf{v}_i \rangle \neq 0 \bmod m$ (i.e., $H \not\subseteq H_i$ and $H_i \not\subseteq H)$ and $\textbf{v} = \sum_{i \in \mathrm{\Omega}} \textbf{v}_i \bmod m$. Compute the identifier for party $P_i$ as: $x_i = \langle \textbf{v}, \textbf{v}_i \rangle \bmod m$. 
	\item \label{St2} For each party $P_e \in \mathcal{P} \setminus \mathrm{\Omega}$, select a unique \textit{covering party} $P_i \in \mathrm{\Omega}$. Let $H_i \in \mathcal{H}$ be the set represented by $P_i$'s covering vector, $\textbf{v}_i \in \mathcal{V}$. Randomly sample $H_j \in \mathcal{H}$, such that, $H_i \subset H_j$. Let $\textbf{v}_j \in \mathcal{V}$ be the covering vector representing $H_j$. 
	\item Compute $\textbf{v}_e \in \mathcal{V}$ such that: $\textbf{v}_e + \textbf{v}_i = \textbf{v}_j \bmod m$. Verify that $\langle \textbf{v}, \textbf{v}_e \rangle \neq 0 \bmod m$, which translates into $H \not\subseteq H_e, H_e \not\subseteq H$, for $H_e \in \mathcal{H}$ represented by $\textbf{v}_e$. If these requirements do not hold, go back to Step~\ref{St2}. 
	\item Compute the identifier for party $P_e$ as: $x_e = \langle \textbf{v}, \textbf{v}_e \rangle \bmod m$. Generating identifiers in this manner for parties $P_e \in \mathcal{P} \setminus \mathrm{\Omega}$ ensures that they are ``covered'' by the identifiers of parties in $\mathrm{\Omega}$. Since each party $P_i \in \mathrm{\Omega}$ can ``cover'' at most one party $P_e \in \mathcal{P} \setminus \mathrm{\Omega}$, our scheme requires that $|\mathcal{P}| \leq 2 \cdot |\mathrm{\Omega}|$.
	\item Each party $P_z \in \mathcal{P}$ receives an \textit{access structure token} $t^{(\Gamma)}_z = \mu^{x_z} \bmod q$, where $\mu \xleftarrow{\; \$ \;} \mathbb{Z}^*_q \setminus \{1\}$.
\end{enumerate}
In case of an identifier collision, i.e., $x_i = x_j$, where $x_j$ is the identifier of another party $P_j \in \mathcal{P}$, re-generate the identifier for either $P_i$ or $P_j$. Recall from Corollary~\ref{corImp} that $\langle \textbf{v}, \textbf{v}_i \rangle$ occupies $\leq 2^{r}-1$ residue classes modulo $m$. Therefore, the probability of an identifier collision is $\approx 1/(2^r-1)^2 < 1/(2^\ell-1)^2$, which may be non-negligible. Since our scheme works with minimal authorized subsets $\mathrm{\Omega}$ such that $|\mathrm{\Omega}| \geq \lceil \ell/2 \rceil$, it supports $2^{\binom{\ell}{\ell/2+1}}$ out of the $2^{2^{\ell - O(\log \ell)}}$ total monotone access structures over a set of size $\ell$.\\[2mm]
\textbf{Access Structure Identification.} Theorem~\ref{Th} proves that any authorized subset of parties $\mathcal{A} \in \Gamma$ can use its set of access structure tokens, $\{t^{(\Gamma)}_i\}_{i \in \mathcal{A}}$, to identify itself as a member of the access structure $\Gamma$.
\begin{theorem}\label{Th}
	Every authorized subset $\mathcal{A} \in \Gamma$ can identify itself as a member of the access structure $\Gamma$ by verifying that: $\prod_{i \in \mathcal{A}} t^{(\Gamma)}_i = 1 \bmod q$.	
\end{theorem}
\begin{proof}
	Recall that for any authorized subset $\mathcal{A} \in \Gamma$, it holds that the set $H_\mathcal{A} \in \mathcal{H}$, represented by $\sum_{i \in \mathcal{A}} \textbf{v}_i = \textbf{v}_\mathcal{A}$, is a superset of the \textit{access structure set} $H \in \mathcal{H}$, i.e., $H \subseteq H_\mathcal{A}$. Hence, from Theorem~\ref{lemma} and Corollary~\ref{corImp}, it follows that: $\langle \textbf{v}, \textbf{v}_\mathcal{A} \rangle = 0 \bmod m = y \cdot m = y \cdot \varphi(q)$, where $y$ is a positive integer. This translates into $\mu^{\langle \textbf{v}, \textbf{v}_\mathcal{A} \rangle} = 1 \bmod q$ (using Euler's theorem). Hence, the following holds for all authorized subsets $\mathcal{A} \in \Gamma$:
	\[\prod\limits_{i \in \mathcal{A}} t^{(\Gamma)}_i = \prod\limits_{i \in \mathcal{A}} \mu^{x_i} = \mu ^{\left\langle \textbf{v}, \sum\limits_{i \in \mathcal{A}} \textbf{v}_i \right\rangle} = \mu^{\langle \textbf{v}, \textbf{v}_\mathcal{A} \rangle} = \mu^{y \cdot \varphi(q)} = 1 \bmod q. \eqno \qed\]
\end{proof}
\textbf{Perfect Soundness and Computational Hiding.}
\begin{theorem}\label{profL}
	Every unauthorized subset $\mathcal{B} \notin \Gamma$ can identify itself to be outside $\Gamma$ by using its set of access structure tokens, $\{t^{(\Gamma)}_i\}_{i \in \mathcal{B}}$, to verify that: $\prod_{i \in \mathcal{B}} t^{(\Gamma)}_i \neq 1 \bmod q$. Given that the Generalized Diffie-Hellman problem is hard, the following holds for all unauthorized subsets $\mathcal{B} \notin \Gamma$ and all access structures $\Gamma' \subseteq 2^{\mathcal{P}}$, where $\Gamma \neq \Gamma'$ and $\mathcal{B} \notin \Gamma'$:  
	\[\Big|Pr[\Gamma~|~ \{t^{(\Gamma)}_i\}_{i \in \mathcal{B}}] - Pr[\Gamma'~|~ \{t^{(\Gamma)}_i\}_{i \in \mathcal{B}}]\Big| \leq \epsilon(\omega), \]
	where $\omega = |\mathcal{P} \setminus \mathcal{B}|$ is the security parameter and $\epsilon$ is a negligible function.
\end{theorem}
\begin{proof}
		It follows from the ASES procedure that for all unauthorized subsets $\mathcal{B} \notin \Gamma$, it holds that the set $H_\mathcal{B} \in \mathcal{H}$, represented by $\sum_{i \in \mathcal{B}} \textbf{v}_i = \textbf{v}_\mathcal{B}$, cannot be a superset or subset of the \textit{access structure set} $H \in \mathcal{H}$. Hence, it follows from Theorem~\ref{lemma} and Corollary~\ref{corImp} that: $\langle \textbf{v}, \textbf{v}_\mathcal{B} \rangle \neq 0 \bmod m$, which translates into the following relation by Euler's theorem (since $m = \varphi(q)$ and $\mu \xleftarrow{\; \$ \;} \mathbb{Z}^*_q \setminus \{1\}):$
		\[
		\prod_{i \in \mathcal{B}} t^{(\Gamma)}_i = \prod_{i \in \mathcal{B}} \mu^{x_i} = \mu^{\left\langle \textbf{v}, \sum\limits_{i \in \mathcal{B}} \textbf{v}_i \right\rangle} = \mu^{\langle \textbf{v}, \textbf{v}_\mathcal{B} \rangle} \neq 1 \bmod q.\]
		Hence, any unauthorized subset $\mathcal{B} \notin \Gamma$ can identify itself as not being a part of the access structure $\Gamma$ by simply multiplying its \textit{access structure tokens}, $\{t^{(\Gamma)}_i\}_{i \in \mathcal{B}}$. The security parameter $\omega = |\mathcal{P} \setminus \mathcal{B}|$ accounts for this minimum information that is available to any unauthorized subset $\mathcal{B} \notin \Gamma$.
		
		If some unauthorized subset $\mathcal{B} \notin \Gamma$ has non-negligible advantage in distinguishing access structure $\Gamma$ from any other $\Gamma' \subseteq 2^{\mathcal{P}}$, where $\Gamma \neq \Gamma'$ and $\mathcal{B} \notin \Gamma'$, then the following must hold for some non-negligible function $\chi$:
		\begin{equation}\label{proofEq}
		\left|Pr[\Gamma~|~ \{t^{(\Gamma)}_i\}_{i \in \mathcal{B}}] - Pr[\Gamma'~|~ \{t^{(\Gamma)}_i\}_{i \in \mathcal{B}}]\right| \geq \chi(\omega),
		\end{equation}
		
		Let $g \in \mathbb{Z}^*_q$ be a generator of $\mathbb{Z}^*_q$ (recall that $\mathbb{Z}_q^*$ is a cyclic group). We know that the setup procedure used to generate $q$ ensures that: $|\mathbb{Z}_q^*| = \varphi(q) \gg |\mathcal{P}|$. Hence, given that $g$ is a generator of $\mathbb{Z}_q^*$, it follows that for each identifier $x_i$, there exists some $a_i \in \mathbb{Z}$ such that: $\mu^{x_i} = g^{a_i} \bmod q.$ Therefore, by extension, it follows that for all sets $\mathcal{B}$, there exists set(s) of $n$ different integers $I_\mathcal{B} = \{a_1, \dots, a_n\}$, where $n = |\mathcal{B}|$, such that $\mu^{\sum_{i \in \mathcal{B}} x_i} = g^{\prod_{i=1}^n a_i} \bmod q$. Hence, it holds that: $\prod_{i \in \mathcal{B}} \mu^{x_i} = g^{\prod_{i=1}^n a_i} \bmod q.$
		
		We know that each unauthorized subset $\mathcal{B} \notin \Gamma$ has at least one proper superset $\mathcal{A} \supsetneq \mathcal{B}$, such that $\mathcal{A} \in \Gamma$. Since $g$ is a generator of $\mathbb{Z}^*_q$, there exists set(s) of $n'$ different integers $I_\mathcal{A} = I_\mathcal{B} \cup \{a_{n+1}, \dots, a_{n'}\}$, where $n' = |\mathcal{A}|$, such that the following holds: $
		\prod_{i \in \mathcal{A}} \mu^{x_i} = g^{\prod_{i=1}^{n'} a_i} \bmod q.$
		
		We know that in order to satisfy Equation~\ref{proofEq}, $\mathcal{B}$ must gain some non-negligible information about $g^{\prod_{i=1}^{n'} a_i}$ in $\mathbb{Z}^*_q$. We also know that $\mathcal{B}$ can compute $g^{\prod_{i=1}^{n} a_i} \bmod q$. Hence, it follows directly from Definition~\ref{GDH} that gaining any non-negligible information about $g^{\prod_{i=1}^{n'} a_i}$ from $g^{\prod_{i=1}^{n} a_i}$ in $\mathbb{Z}^*_q$ requires solving the Generalized Diffie-Hellman (GDH) problem. Therefore, Equation~\ref{proofEq} cannot hold given that the GDH assumption holds. Hence, the advantage of $\mathcal{B} \notin \Gamma$ must be negligible in the security parameter $\omega$. $\qed$
\end{proof}

\subsection{Building the Full Scheme}\label{Full} The following procedure allows an honest dealer to employ the ASES scheme and realize an access structure hiding computational secret sharing scheme.
\begin{enumerate}
	\item Perform ASES to generate \textit{access structure tokens} $t^{(\Gamma)}_z = \mu^{x_z} \bmod q$, for each party $P_z \in \mathcal{P}$.
	\item Follow Step 1 of the setup procedure of ASES to generate a suitable prime $q'$.   
	\item Generate a set-system $\mathcal{H}'$ modulo $m'$ (as defined by Theorem~\ref{lemma}), where $m' = \varphi(q')$. Let $\mathcal{V'}$ denote the covering vector family (as defined by Corollary~\ref{corImp}) that is formed by the representative vectors $\textbf{v}_i \in \mathcal{V}$ for the sets $H_i \in \mathcal{H}$. 
	\item Generate the secret that needs to be shared: $k \xleftarrow{\; \$ \;} \mathbb{Z}^*_{q'}$, and randomly sample $|\mathrm{\Omega}|$ integers, $\{b_i\}_{i=1}^{|\mathrm{\Omega}|}$, such that: $\prod_{i=1}^{|\mathrm{\Omega}|} b_i = k \bmod q'$.
	\item Generate $\gamma \xleftarrow{\; \$ \;} \mathbb{Z}^*_{q'} \setminus \{1\}$. For each party $P_j \in \mathcal{P} \setminus \mathrm{\Omega}$, employ ASES with parameters $\{m', q', \mathcal{H'}, \mathcal{V'}, \gamma\}$ to generate identifier $y_j \in \mathbb{Z}_m$, and access structure token: $s^{(k)}_j = \gamma^{y_j} \bmod q'$. Party $P_j$ receives $s^{(k)}_j$ as its share.
	\item The share for each party $P_i \in \mathrm{\Omega}$ is generated as: $s^{(k)}_i = (b_i \cdot \gamma^{y_i}) \bmod q'$. Each party $P_z \in \mathcal{P}$ receives <access structure token, share> pair: $(t^{(\Gamma)}_z, s^{(k)}_z)$.
\end{enumerate}
\textbf{Completeness, Soundness, Correctness, Secrecy and Hiding:} 
We prove that our access structure hiding computational secret sharing scheme satisfies the completeness, soundness, correctness, hiding and secrecy requirements outlined by the definition of Access Structure Hiding Computational Secret Sharing (see Definition~\ref{MainDef}). Since independent iterations of ASES are used to generate the access structure tokens and shares, perfect completeness follows directly from Theorem~\ref{Th}. Similarly, perfect soundness and computational hiding follow directly from Theorem~\ref{profL}. Hence, we move on to proving perfect correctness and computational secrecy. \\[2mm] 
\textit{Perfect Correctness:} It follows directly from Theorem~\ref{Th} that for all authorized subsets $\mathcal{A} \in \Gamma$, it holds that: $\prod_{i \in \mathcal{A}} \gamma^{y_i} = 1 \bmod q'$. Hence, any $\mathcal{A} \in \Gamma$ can reconstruct the secret, $k$, by combining its shares as:  
\[\prod\limits_{i \in \mathcal{A}} s^{(k)}_i \bmod q' = 1 \cdot \prod\limits_{y \in \mathrm{\Omega}} b_y \bmod q' = k, \eqno \text{(using }\prod\limits_{i \in \mathcal{A}} \gamma^{y_i} = 1 \bmod q').\]
\begin{theorem}
	The maximum share size of our access structure hiding secret sharing scheme for any access structure is $(1+ o(1)) \dfrac{2^{\ell+1}}{\sqrt{\pi \ell/2}}$.
\end{theorem}
\begin{proof}
	Our access structure hiding secret sharing scheme is designed to ``encode'' minimal authorized subsets. It is easy to verify that the maximum number of unique minimal authorized subsets in any access structure is $\binom{\ell}{\ell/2}$. For each minimal authorized subset, each party $P_z \in \mathcal{P}$ receives two elements, $s^{(k)}_z$ and $t^{(\Gamma)}_z$, both of which have (almost) the same size as the secret. Hence, it follows that the maximum share size for any (supported) access structure is: 
	\begin{align*}
	\max \left(\mathrm{\Pi}^{(k)}\right) &\approx \binom{\ell}{\ell/2}2|k|\\ &= (1+ o(1)) \dfrac{2^{\ell+1}}{\sqrt{\pi \ell/2}}|k|, \qquad \text{(using results from~\cite{Das[20]})}.
	\end{align*}
	Hence, the maximum share size with respect to the secret size $|k|$ is: $(1+ o(1)) \dfrac{2^{\ell+1}}{\sqrt{\pi \ell/2}}. \qed$
\end{proof} 	

\textit{Computational Secrecy:} Since independent iterations of ASES are used to generate the sets $\{t^{(\Gamma)}_z\}_{z \in \mathcal{P}}$ and $\{s^{(k)}_z\}_{z \in \mathcal{P}}$, computational indistinguishability (w.r.t. security parameter $\omega = |\mathcal{P} \setminus \mathcal{B}|)$ of all different access structures $\Gamma, \Gamma' \subseteq 2^{\mathcal{P}}$, for all unauthorized subsets $\mathcal{B} \notin \Gamma, \Gamma'$ follows directly from Theorem~\ref{profL}, i.e., it holds that: $$\Big|Pr[\Gamma~|~ \{t^{(\Gamma)}_i\}_{i \in \mathcal{B}}, \{s^{(k)}_i\}_{i \in \mathcal{B}}] - Pr[\Gamma'~|~ \{t^{(\Gamma)}_i\}_{i \in \mathcal{B}}, \{s^{(k)}_i\}_{i \in \mathcal{B}}]\Big| \leq \epsilon(\omega).$$ 
\begin{theorem}\label{ThMM}
	Given that GDH problem is hard, it holds for every unauthorized subset $\mathcal{B} \notin \Gamma$ and all different secrets $k_1, k_2 \in \mathcal{K}$ that the distributions $\{s_i^{(k_1)}\}_{i \in \mathcal{B}}$ and $\{s_i^{(k_2)}\}_{i \in \mathcal{B}}$ are computationally indistinguishable w.r.t. the security parameter $\omega = |\mathcal{P} \setminus \mathcal{B}|$.
\end{theorem}
\begin{proof}
	Since the set $\{b_i\}_{i=1}^{|\mathrm{\Omega}|}$ is generated randomly, secrecy of the $b_i$ values follows from one-time pad. Moving on to the secrecy of $\gamma^{y_i}$ values: since $\gamma(\neq 1)$ is a random element from $\mathbb{Z}_{q'}^*$, there exists a generator $g$ of $\mathbb{Z}_{q'}^*$ (note that $\mathbb{Z}_{q'}^*$ is a cyclic group) such that for each identifier, $y_i$, generated by the ASES procedure, there exists an $a_i \in \mathbb{Z}$ such that: $\gamma^{y_i} = g^{a_i} \bmod q'$. By extension, there exists set(s) of $n$ different integers $I_\mathcal{B} = \{a_1, \dots, a_n\}$, where $n = |\mathcal{B}|$, such that: $\prod_{i \in \mathcal{B}} \gamma^{y_i} = g^{\prod_{i=1}^n a_i} \bmod q'$. We know that each unauthorized subset $\mathcal{B} \notin \Gamma$ has at least one proper superset $\mathcal{A} \supsetneq \mathcal{B}$, such that $\mathcal{A} \in \Gamma$. Since $g$ is a generator of $\mathbb{Z}^*_{q'}$, there exists set(s) of $n'$ different integers $I_\mathcal{A} = I_\mathcal{B} \cup \{a_{n+1}, \dots, a_{n'}\}$, where $n' = |\mathcal{A}|$, such that: $\prod_{i \in \mathcal{A}} \gamma^{y_i} = g^{\prod_{i=1}^{n'} a_i} \bmod q'.$ It follows from Definition~\ref{GDH} that in order to gain any non-negligible information about $g^{\prod_{i=1}^{n'} a_i}$ from $g^{\prod_{i=1}^{n} a_i}$, in $\mathbb{Z}_{q'}^*$, $\mathcal{B}$ must solve the GDH problem. Therefore, for every unauthorized subset $\mathcal{B} \notin \Gamma$, computational indistinguishability of $\{s_i^{(k_1)}\}_{i \in \mathcal{B}}$ and $\{s_i^{(k_2)}\}_{i \in \mathcal{B}}$ w.r.t. the security parameter $\omega$ follows directly from the GDH assumption.   $\qed$
\end{proof}	 
\section*{Open Problems} Our access structure hiding secret sharing scheme requires that $|\mathcal{P}| \leq 2 |\mathrm{\Omega}|$, where $\mathrm{\Omega} \in \Gamma_0$ is any minimal authorized subset. It is worth exploring whether this restriction can be further relaxed, or removed. Another interesting problem is defining and constructing set-systems and vector families that can support simultaneous encoding of multiple minimal authorized subsets.

\bibliographystyle{plainurl}
\bibliography{CO2020}
\end{document}